\documentclass[conference]{IEEEtran}
\IEEEoverridecommandlockouts
\usepackage[letterpaper, 
top=0.75in, 
bottom=1.1in, 
left=0.64in, 
right=0.64in, 
columnsep=0.3in]{geometry}

\usepackage{lipsum} 

\usepackage{cite}
\usepackage{amsmath,amssymb,amsfonts}
\usepackage{algorithmic}
\usepackage{graphicx}
\usepackage{textcomp}
\usepackage{xcolor}
\usepackage[colorlinks]{hyperref}
\usepackage{amsmath}
\usepackage{amsfonts}
\usepackage{amssymb}
\usepackage{amsthm}
\usepackage{bm}
\usepackage{tabularx}
\usepackage{mathrsfs} 
\usepackage{algorithmic}
\usepackage{algorithm} 
\usepackage{mathtools}
\usepackage{url}
\usepackage{hyperref}

\newtheorem{Proposition}{Proposition}
\newtheorem{lemma}{Lemma}

\usepackage{stfloats}
\usepackage{float}
\usepackage{graphicx}
\usepackage{xcolor}
\usepackage{subfigure}

\makeatletter
\def\blfootnote{\xdef\@thefnmark{}\@footnotetext}
\makeatother
\def\BibTeX{{\rm B\kern-.05em{\sc i\kern-.025em b}\kern-.08em
		T\kern-.1667em\lower.7ex\hbox{E}\kern-.125emX}}

\begin{document}
\title{\LARGE Secrecy Performance Analysis of RIS-Aided Fluid Antenna Systems}

\author{
	\IEEEauthorblockN{Farshad Rostami Ghadi\textsuperscript{†}\hspace{-1mm}, Kai-Kit Wong\textsuperscript{†}\hspace{-1mm}, Masoud Kaveh\textsuperscript{‡}, F. Javier Lopez-Martinez\textsuperscript{§}\hspace{-1mm}, Wee Kiat New\textsuperscript{†}\hspace{-1mm}, and Hao Xu\textsuperscript{†}\hspace{-1mm}}
	\IEEEauthorblockA{\textsuperscript{†}Department of Electronic and Electrical Engineering, University College London, London, United Kingdom. \\}
	\IEEEauthorblockA{\textsuperscript{‡}Department of Information and Communication Engineering, Aalto University, Espoo, Finland. \\}
	\IEEEauthorblockA{\textsuperscript{§}Department of Signal Theory, Networking and Communications, University of Granada, Granada, Spain. \\}
				Emails: \{f.rostamighadi, kai-kit.wong, a.new, hao.xu\}@ucl.ac.uk; masoud.kaveh@aalto.fi; fjlm@ugr.es.
}

\maketitle
\begin{abstract}
This paper examines the impact of emerging fluid antenna systems (FAS) on reconfigurable intelligent surface (RIS)-aided secure communications. Specifically, we consider a classic wiretap channel, where a fixed-antenna transmitter sends confidential information to an FAS-equipped legitimate user with the help of an RIS, while an FAS-equipped eavesdropper attempts to decode the message. To evaluate the proposed wireless scenario, we first introduce the cumulative distribution function (CDF) and probability density function (PDF) of the signal-to-noise ratio (SNR) at each node, using the central limit theorem and the Gaussian copula function. We then derive a compact analytical expression for the secrecy outage probability (SOP). Our numerical results reveal how the incorporation of FAS and RIS can significantly enhance the performance of secure communications.
\end{abstract}


%
\vspace{0cm}\section{Introduction}\label{sec-intro} 
The fluid antenna system (FAS) has recently emerged as a groundbreaking technology for the next generation of wireless communication, a.k.a., sixth-generation (6G), leveraging its position-flexibility as a novel degree of freedom (DoF) to optimize diversity and multiplexing gains \cite{new2024tutorial}. Compared to a traditional antenna system, where the antenna location is fixed, FAS encompasses liquid-based antennas \cite{huang2021liquid}, reconfigurable radio-frequency (RF) pixels \cite{rodrigo2014frequency}, or movable mechanical antenna structures \cite{basbug2017design} designed to modify their shape and position, thereby reconfiguring their radiation characteristics \cite{ghadi2023copula}. Notably, a fluid antenna can switch positions (i.e., ports) within a designated small area, a distinct capability that is highly beneficial for mobile devices where physical limitations restrict antenna deployment options. The concept of FAS was initially introduced to wireless communication systems in 2020 \cite{wong2020fluid}, and its application to the forthcoming 6G technology has since been explored in various studies \cite{shojaeifard2022mimo,new2023information,ghadi2024cache,new2023fluid,xu2023channel,ghadi2024isac,xu2023capacity,zhang2024pixel,shen2024design}.

Independently, the reconfigurable intelligent surfaces (RIS) have emerged as a transformative technology with the potential to revolutionize coverage and performance in advanced wireless communication systems over recent years \cite{basar2019wireless}. RIS are engineered surfaces embedded with numerous low-cost, electronically controlled reflective elements that can dynamically manipulate the wireless signal environment. By adjusting these elements, RIS can intelligently redirect radio waves from base stations (BSs) towards targeted mobile users, thus optimizing signal strength and extending coverage \cite{liu2021reconfigurable}. The promise of RIS is bolstered by significant strides in low-power electronics, which have made their continuous and efficient operation more practical in 6G wireless technology \cite{pan2021reconfigurable}. 

 On the other hand, secure and reliable communications are paramount in 6G wireless communication due to the increasing reliance on interconnected devices and critical infrastructure. In this regard, physical layer security (PLS) emerges as a vital method, leveraging inherent properties of the communication channel to provide robust, low-complexity security solutions that protect data from eavesdropping, ensuring the confidentiality of transmitted information. However, the current beamforming methods in PLS rely on fixed-position antenna arrays \cite{xiao2023array}, which can raise significant challenges due to their inherent limitations in flexibility, adaptability, and security. Therefore, it is apparent that 6G wireless technology needs to incorporate state-of-the-art technologies to address these challenges in secure communication.

Despite RIS, FAS, and PLS complementing one another, their synergistic effects are not yet fully grasped. However, only recently, the integration of FAS with an optimized RIS was examined in \cite{ghadi2024performance}, where the distribution of the equivalent channel for a FAS-equipped mobile user was derived, along with the outage probability and delay outage rate. Additionally, FAS-aided communication in PLS was studied in \cite{ghadi2024physical}, where key secrecy metrics such as secrecy outage probability (SOP), average secrecy capacity, and secrecy energy efficiency were analyzed. Furthermore, the performance of the average secrecy rate in secure communication, when only the legitimate user is equipped with FAS, was investigated in \cite{tang2023fluid}. These notable  contributions have highlighted the superiority of utilizing FAS compared to traditional fixed-antenna systems. Therefore, unlike the previous works, we consider a RIS-aided FAS in secure communication, where a fixed-antenna transmitter (Alice) conveys confidential message to a legitimate receiver (Bob) with the assist of an optimized RIS, while an eavesdropper (Eve) attempts to intercept  the information. Moreover, it is assumed that Bob and Eve are equipped with a one-sided planar FAS, which can switch to  the optimal position within a designated  two-dimensional (2D) space for reception. 

The key technical contributions of this work are as follows: (i) We introduce the cumulative distribution function (CDF) and probability density function (PDF) of the received signal-to-noise ratio (SNR), utilizing the central limit theorem (CLT) and the Gaussian copula function; (ii) We then derive a compact analytical expression for the SOP using the numerical Gaussian-Laguerre quadrature (GLQ) integration method; (iii) Finally, our numerical results demonstrate that deploying FAS in RIS-aided communication significantly enhances system performance, leading to more secure and reliable transmission.  
\section{System Model}\label{sec-sys}
We consider a secure wireless communication system as illustrated in Fig. \ref{fig-model}, where a transmitter (Alice) wants to send confidential information $x$ with total transmit power $P$ to a legitimate mobile user (Bob) through a RIS, while an eavesdropper (Eve) endeavors to decode the message from its received signal. To facilitate notation, we define the subscript $i$ to indicate the node associated with Bob and Eve, i.e., $i\in\left\{\mathrm{b,e}\right\}$, and we use subscript "$\mathrm{r}$" for the RIS. Without loss of generality, we assume that Alice has a single fixed-position antenna and the RIS includes $M$ reflecting elements, while a planar FAS with $N_i$ predefined positions (i.e., ports) distributed over an area of $W_i$ is deployed at both Bob and Eve. In particular, we assume that the planar FAS has a grid structure, with $N_{l,i}$ ports evenly spaced along a linear distance of $W_{l,i}\lambda$ for $l\in\{1,2\}$ such that $N_i=N_{1,i}\times N_{2,i}$ and $W_i=\lambda^2\left(W_{1,i},W_{2,i}\right)$, where $\lambda$ denotes the wavelength of the carrier frequency. Furthermore, to convert the 2D indices to a 1D index, we use a mapping function defined as $\mathcal{F}\left(n_i\right)=\left(n_{1,i}, n_{2,i}\right)$, where $n_i\in\left\{1,\dots,N_i\right\}$ and $n_{l,i}\in\left\{1,\dots,N_{l,i}\right\}$. Therefore, the received signal at node $i$ is expressed as
\begin{align}
	y_{i}=\bm{g}_\mathrm{r}^T\mathbf{\Psi}\bm{h}_{i}x+z_i=h_{\mathrm{eq},i}x+z_i,
\end{align}
in which $h_{\mathrm{eq},i}$ denotes the equivalent channel from Alice to $n_i$-th port at node $i$ through RIS. The vectors $\bm{g}_\mathrm{r}=\left[g_1\mathrm{e}^{-j\theta_1},\dots,g_M\mathrm{e}^{-j\theta_M}\right]^T\in\mathbb{C}^{M\times 1}$ and $\bm{h}_{i}=\left[h_{1,n_i}^i\mathrm{e}^{-j\beta_1^i},\dots,h_{M,n_i}^i\mathrm{e}^{-j\beta^i_M}\right]^T\in\mathbb{C}^{M\times 1}$ consist of the channel coefficients from Alice to RIS and from RIS to $n_i$-th port at node $i$, respectively. Besides, $\theta_m$ and $\beta_m^i$ indicate the phases of the respective channel coefficients, whereas $g_m$ and $h_{m,n_i}^i$ denote the amplitudes of $\bm{g_r}$ and $\bm{h}_i$, respectively. Additionally, the diagonal matrix $\mathbf{\Psi}=\text{diag}\left(\left[\zeta_1\mathrm{e}^{j\psi_1},\dots,\zeta_M\mathrm{e}^{j\psi_M}\right]\right)\in\mathbb{C}^{M\times M}$ includes the adjustable phases set by the reflecting elements of the RIS.
 Furthermore, $z_i$ represents the independent identically distributed (i.i.d.) additive white Gaussian noise (AWGN) with zero mean and variance $\sigma_i^2$ at each FAS port of node $i$. 

In FAS, ports can move freely and be positioned arbitrarily close to each other, leading to spatial correlation. Therefore, assuming a planar FAS with a half-space reception coverage in front, the spatial correlation between any two arbitrary ports, $n_i=\mathcal{F}^{-1}\left(n_{1,i},n_{2,i}\right)$ and $\tilde{n}_i=\mathcal{F}^{-1}\left(\tilde{n}_{1,i},\tilde{n}_{2,i}\right)$, can be determined as \cite{ghadi2024performance}
\begin{align}\notag
&\varrho_{n_i,\tilde{n}_i}=\\
&{\rm sinc}\left(\frac{2}{\lambda}\sqrt{\left(\frac{|n_{1,i}-\tilde{n}_{1,i}|}{N_{1,i}-1}W_{1,i}\right)^2+\left(\frac{|n_{2,i}-\tilde{n}_{2,i}|}{N_{2,i}-1}W_{2,i}\right)^2}\right), \label{eq-cov}
\end{align}
 where $\tilde{n}_i\in\left\{1,\dots,N_i\right\}$,  $\tilde{n}_{l,i}\in\left\{1,\dots,N_{l,i}\right\}$, and $\mathrm{sinc(t)=\frac{\sin(\pi t)}{\pi t}}$ denotes the sinc function.  Therefore, the spatial correlation matrix $\mathbf{R}_i$ for node $i$ is defined as
\begin{align}
\mathbf{R}_i=\begin{bmatrix}
 	\varrho_{1,1} & \varrho_{1,2} &\dots& \varrho_{1,N_i}\\
 	\varrho_{2,1} & \varrho_{2,2} &\dots& \varrho_{2,N_i}\\ \vdots & \vdots & \ddots & \vdots\\
 	\varrho_{N_i,1} & \varrho_{N_i,2} &\dots& \varrho_{N_i,N_i}
\end{bmatrix}.
\end{align}
Further, by leveraging the fact that FAS can activate an optimal port that maximizes the SNR for communication, the received SNR at node $i$ can be defined as
\begin{align}\label{eq-snr1}
	\gamma_i=\frac{P\left|\left[h_{\mathrm{eq},i}\right]_{n^*_i}\right|^2}{\sigma^2_i\left(d_\mathrm{r}d_i\right)^\alpha}=\overline{\gamma}_i\textcolor{black}{\left|\left[h_{\mathrm{eq},i}\right]_{n^*_i}\right|^2},
\end{align} 
where $\overline{\gamma}_i$ indicates the average SNR, with $d_\mathrm{r}$ and $d_i$ representing distances between Alice and RIS, and between the RIS to node $i$, respectively, and $\alpha>2$ is the path-loss exponent. Besides, $n^*_i$  specifies the index of the port that has been optimally selected at node $i$, i.e., 
\begin{align}
n^*_i=\arg\underset{n_i}{\max}\left\{\left|\left[h_{\mathrm{eq},i}\right]_{n_i}\right|^2\right\},
\end{align}
where the notation $\left[h_{\mathrm{eq},i}\right]_{n_i}$ denotes the $n_i$-th entry of $h_{\mathrm{eq},i}$. Thus, the channel gain at node $i$, incorporating the effects of FAS, is given by
\begin{align}
h_{\mathrm{fas},i}^2=\max\left\{|h^i_\mathrm{eq,1}|^2,|h^i_\mathrm{eq,2}|^2,\dots,|h^i_{\mathrm{eq},{N_i}}|^2\right\}. \label{eq-hfas}
\end{align}
Therefore, \eqref{eq-snr1} can be rewritten as
\begin{align}
	\gamma_i = \overline{\gamma}_i h_{\mathrm{fas},i}^2.\label{eq-snr}
\end{align}
\begin{figure}[!t]
	\centering
	\includegraphics[width=0.9\columnwidth]{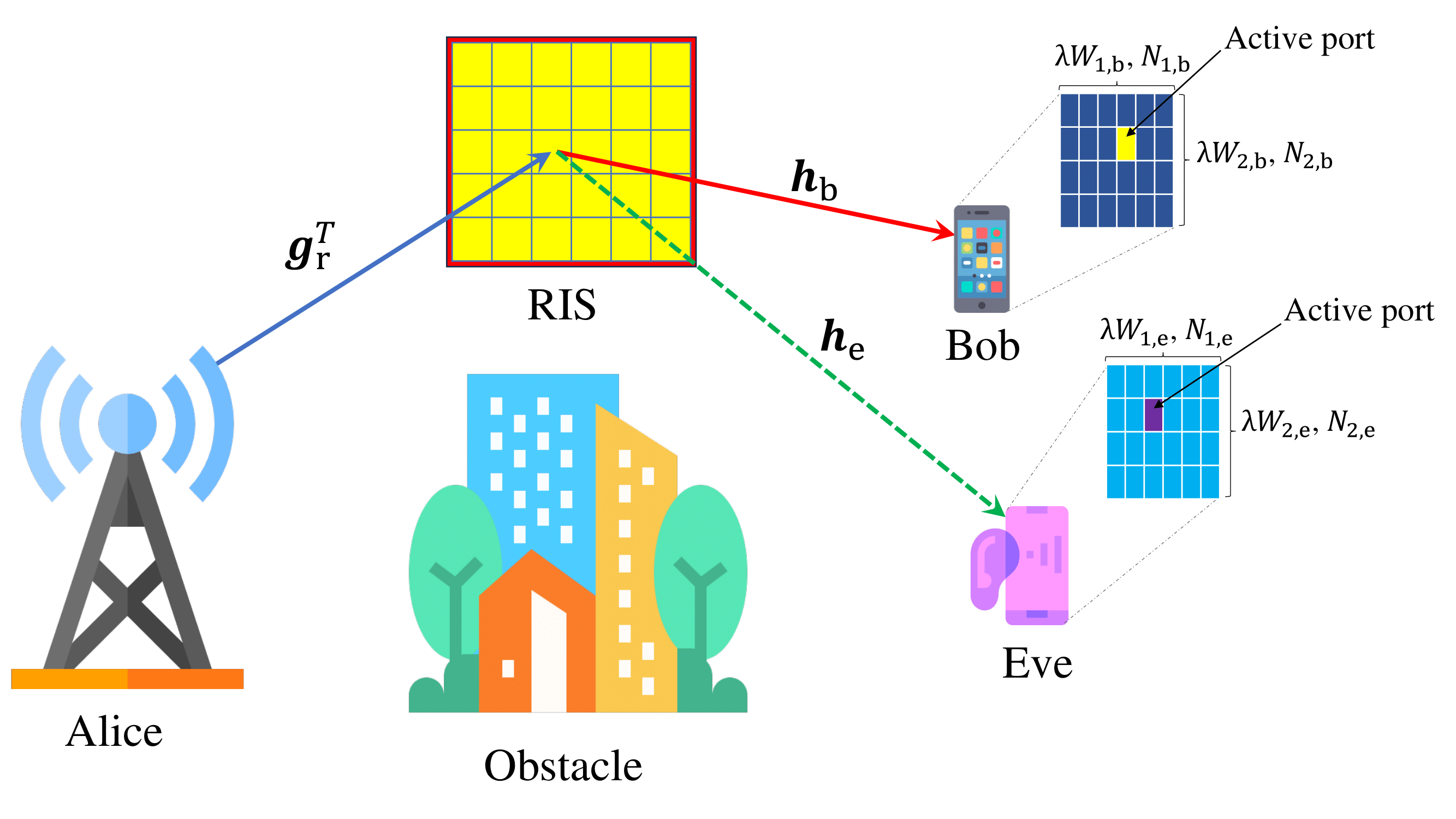}
	\caption{A RIS-aided wiretap channel involving the single fixed-antenna Alice communicating with the FAS-equipped Bob and the FAS-equipped eavesdropper Eve.}\vspace{0cm}\label{fig-model}
\end{figure}
\section{Performance Analysis}
In this section, we first derive the CDF and PDF of the SNR at node $i$, utilizing the CLT and copula theory. Next, we present a compact series-form expression for the SOP by exploiting the numerical GLQ integration technique.
\subsection{Statistical Characterization}
\subsubsection{Distribution of $\gamma_\mathrm{b}$} Given  \eqref{eq-hfas} and \eqref{eq-snr}, we first need to extend the equivalent channel gain at $n_\mathrm{b}$-th port Bob as follows
\begin{align}
	\left|h_{\mathrm{eq},n_\mathrm{b}}\right|^2&=\left|\sum_{m=1}^Mg_mh_{m,n_\mathrm{b}}^\mathrm{b}\mathrm{e}^{-j(\psi_m-\theta_m-\beta_m^\mathrm{b})}\right|^2\\
	&\hspace{-.5mm}\overset{(a)}{=}\left(\sum_{m=1}^Mg_mh_{m,n_\mathrm{b}}^\mathrm{b}\right)^2=A^2,
\end{align}
where $(a)$ stems from assuming perfect channel state information (CSI) for the RIS configuration, which allows for an ideal phase adjustment, i.e., $\psi_m = \theta_m + \beta_m^\mathrm{b}$. Assuming a large number of reflecting elements, we can use the CLT to approximate the distribution of $A^2$. Consequently, $A = \sum_{m=1}^M g_m h_{m,n_\mathrm{b}}^\mathrm{b}$ is well approximated as a Gaussian random variable (RV) with mean $\mu_A = \frac{M\pi}{4}$ and variance $\sigma^2_A = M \left(1 - \frac{\pi^2}{16}\right)$. Therefore, $A^2$ follows a non-central chi-square distribution with one degree of freedom such that the corresponding CDF and PDF are respectively given by 
\begin{align}
	F_{A^2}(a) = 1-Q_{\frac{1}{2}}\left(\sqrt{\frac{\tau}{\sigma_A^2}},\sqrt{\frac{a}{\sigma_A^2}}\right),\label{cdf-a}
\end{align}
and
\begin{align}
	f_{A^2}(a) = \frac{1}{2\sigma^2_\mathrm{A}}\left(\frac{a}{\tau}\right)^{-\frac{1}{4}}\mathrm{exp}\left(-\frac{a+\tau}{2\sigma^2}\right)\mathcal{I}_{-\frac{1}{2}}\left(\frac{\sqrt{a\tau}}{\sigma_A^2}\right),\label{eq-pdf-a}
\end{align}
where $\tau=\mu_\mathrm{A}^2$ denotes the non-centrality parameter, $Q_l(v,w)$ represents the Marcum $Q$-function with order $l$, and $\mathcal{I}_\nu\left(\cdot\right)$ indicates the modified Bessel function of the first kind of order $\nu$. Given the distribution of the equivalent channel gain, the distribution of $\gamma_\mathrm{b}$ can be characterized as follows in the ensuing proposition.
\begin{Proposition}
The CDF and PDF of $\gamma_\mathrm{b}$ are given by \eqref{eq-cdf-gb} and \eqref{eq-pdf-gb}, respectively, where $\boldsymbol{\mathbf{\varphi}}^{-1}_{A^2}$ is defined in \eqref{eq-phi} (see the top of the next page).
\begin{figure*}[t]
	\normalsize
	\begin{align}
		&F_{\gamma_\mathrm{b}}\left(\gamma_\mathrm{b}\right)=\Phi_{\mathbf{R}_\mathrm{b}}\left(\sqrt{2}\mathrm{erf}^{-1}\left(1-2Q_{\frac{1}{2}}\left(\sqrt{\frac{\tau}{\sigma_A^2}},\sqrt{\frac{\gamma_\mathrm{b}}{\sigma_A^2\overline{\gamma}_\mathrm{b}}}\right)\right),\dots,\sqrt{2}\mathrm{erf}^{-1}\left(1-2Q_{\frac{1}{2}}\left(\sqrt{\frac{\tau}{\sigma_A^2}},\sqrt{\frac{\gamma_\mathrm{b}}{\sigma_A^2\overline{\gamma}_\mathrm{b}}}\right)\right);\vartheta_\mathrm{b}\right)\label{eq-cdf-gb}
	\end{align}
	\hrulefill
	\begin{align}\label{eq-pdf-gb}
		f_{\gamma_\mathrm{b}}\left(\gamma_\mathrm{b}\right)=\frac{1}{\overline{\gamma}_\mathrm{b}}\left[\frac{1}{2\sigma^2_\mathrm{A}}\left(\frac{\gamma_\mathrm{b}}{\tau\overline{\gamma}_\mathrm{b}}\right)^{-\frac{1}{4}}\mathrm{exp}\left(-\frac{\gamma_\mathrm{b}+\tau\overline{\gamma}_\mathrm{b}}{2\sigma^2\overline{\gamma}_\mathrm{b}}\right)\mathcal{I}_{-\frac{1}{2}}\left(\frac{\sqrt{\gamma_\mathrm{b}\tau}}{\sigma_A^2\sqrt{\overline{\gamma}_\mathrm{b}}}\right)\right]^{N_\mathrm{b}}
		\frac{\exp\left(-\frac{1}{2}\left(\boldsymbol{\mathbf{\varphi}}^{-1}_{A^2}\right)^T\left(\mathbf{R}_\mathrm{b}^{-1}-\mathbf{I}\right)\boldsymbol{\mathbf{\varphi}}^{-1}_{A^2}\right)}{\sqrt{{\rm det}\left(\mathbf{R}_\mathrm{b}\right)}}
	\end{align}
	\hrulefill
	\begin{align}\label{eq-phi}
		\boldsymbol{\mathbf{\varphi}}^{-1}_{A^2}=\left[\sqrt{2}\mathrm{erf}^{-1}\left(1-2Q_{\frac{1}{2}}\left(\sqrt{\frac{\tau}{\sigma_A^2}},\sqrt{\frac{\gamma_\mathrm{b}}{\sigma_A^2\overline{\gamma}_\mathrm{b}}}\right)\right),\dots,\sqrt{2}\mathrm{erf}^{-1}\left(1-2Q_{\frac{1}{2}}\left(\sqrt{\frac{\tau}{\sigma_A^2}},\sqrt{\frac{\gamma_\mathrm{b}}{\sigma_A^2\overline{\gamma}_\mathrm{b}}}\right)\right)\right]^T
	\end{align}
	\hrulefill
\end{figure*}
	\end{Proposition}
	\begin{proof}
Given \eqref{eq-snr}, the CDF of $\gamma_\mathrm{b}$ can be mathematically expressed as follows
\begin{align}
	F_{\gamma_\mathrm{b}}\left(\gamma_\mathrm{b}\right)&=\Pr\left(\overline{\gamma}_\mathrm{b} h_{\mathrm{fas,b}}^2\leq\gamma_\mathrm{b}\right)\\
	&=\Pr\left(\max\left\{|h^i_\mathrm{eq,1}|^2,\dots,|h^i_{\mathrm{eq},{N_\mathrm{b}}}|^2\right\}\leq \frac{\gamma_\mathrm{b}}{\overline{\gamma}_\mathrm{b}}\right)\\
	&=\Pr\left(\max\left\{A^2,\dots, A^2\right\}\leq \frac{\gamma_\mathrm{b}}{\overline{\gamma}_\mathrm{b}}\right)\\
	&=F_{A^2,\dots,A^2}\left(\frac{\gamma_\mathrm{b}}{\overline{\gamma}_\mathrm{b}},\dots,\frac{\gamma_\mathrm{b}}{\overline{\gamma}_\mathrm{b}}\right)\\
	&\overset{(b)}{=}C\left(F_{A^2}\left(\frac{\gamma_\mathrm{b}}{\overline{\gamma}_\mathrm{b}}\right),\dots,F_{A^2}\left(\frac{\gamma_\mathrm{b}}{\overline{\gamma}_\mathrm{b}}\right);\vartheta\right),\label{eq-cdf1}
\end{align}
where $(b)$ is derived using Sklar's theorem, which allows for the accurate generation of the joint multivariate CDF of arbitrarily correlated RVs by utilizing their respective marginal distributions and a suitable parametric copula $C$ with a dependence parameter $\vartheta$. Specifically, $C\left(\cdot\right):\left[0,1\right]^d\rightarrow \left[0,1\right]$ is a function that defines a joint CDF for $d$ random vectors on the unit cube $\left[0,1\right]^d$ with uniform marginal distributions \cite{ghadi2020copula}, i.e., 
\begin{align}
C\left(u_1,\dots,u_d;\vartheta\right)=\Pr\left(U_1\leq u_1,\dots,U_d\leq u_d\right), \label{eq-c}
\end{align}
where $u_j=F_{S_j}\left(s_j\right)$, with $S_j$ denoting any arbitrary RV for $j\in\left\{1,\dots,d\right\}$, and $\vartheta$ 
 measures the linear or non-linear correlation between arbitrarily correlated RVs. Hence, by applying $u_{n_\mathrm{b}}=F_{A^2}\left(\frac{\gamma_\mathrm{b}}{\overline{\gamma}_\mathrm{b}}\right)$ into \eqref{eq-c}, \eqref{eq-cdf1} is derived. It is noteworthy that \eqref{eq-c} is mathematically valid for any choice of $C$. However, it is understood from \cite{ghadi2023gaussian} that the Gaussian copula can accurately describe the spatial correlation between fluid antenna ports by approximating Jake's model. Therefore, the CDF of $\gamma_\mathrm{b}$ can be derived as \cite{ghadi2023gaussian}
\begin{align}\notag
	&F_{\gamma_\mathrm{b}}\left(\gamma_\mathrm{b}\right)
	=\\
	&\Phi_{\mathbf{R}_\mathrm{b}}\left(\varphi^{-1}\left(F_{A^2}\left(\frac{\gamma_\mathrm{b}}{\overline{\gamma}_\mathrm{b}}\right)\right),\dots,\varphi^{-1}\left(F_{A^2}\left(\frac{\gamma_\mathrm{b}}{\overline{\gamma}_\mathrm{b}}\right)\right);\vartheta_\mathrm{b}\right),\label{eq-cdf-g}
\end{align}
in which $\Phi_{\mathbf{R}_\mathrm{b}}(\cdot)$ indicates the joint CDF of the multivariate normal distribution with zero mean vector and correlation matrix $\mathbf{R}_\mathrm{b}$. Additionally,  $\varphi^{-1}\left(u_j\right)=\sqrt{2}\mathrm{erf}^{-1}\left(2u_j-1\right)$ denotes the quantile function of the standard normal distribution, with $\mathrm{erf}^{-1}\left(\cdot\right)$ being the inverse of the error function $\mathrm{erf}\left(z\right)=\frac{2}{\sqrt{\pi}}\int_0^z\mathrm{e}^{-t^2}dt$. Moreover, $\vartheta_\mathrm{b}$ denotes the dependence parameter of the Gaussian copula at Bob, approximating the correlation coefficient of Jake's model, i.e., $\vartheta_{n_i,\tilde{n}_i}^\mathrm{b}\approx\varrho_{n_i,\tilde{n}_i}^\mathrm{b}$ \cite{ghadi2023gaussian}. Now, by inserting $F_{A^2}(a)$ from \eqref{cdf-a} into \eqref{eq-cdf-g}, the proof of $F_{\gamma_\mathrm{b}}\left(\gamma_\mathrm{b}\right)$ is completed. 

By applying the chain rule to \eqref{eq-cdf1}, the PDF of $\gamma_\mathrm{b}$ can be given by \vspace{0cm}
\begin{align}\notag
	&f_{\gamma_\mathrm{b}}\left(\gamma_\mathrm{b}\right)=\\
	&\frac{\left[f_{A^2}\left(\frac{\gamma_\mathrm{b}}{\overline{\gamma}_\mathrm{b}}\right)\right]^{N_\mathrm{b}}}{\overline{\gamma}_\mathrm{b}} c\left(F_{A^2}\left(\frac{\gamma_\mathrm{b}}{\overline{\gamma}_\mathrm{b}}\right),\dots,F_{A^2}\left(\frac{\gamma_\mathrm{b}}{\overline{\gamma}_\mathrm{b}}\right);\vartheta\right),\label{eq-gen-pdf}
\end{align}
where $c\left(\cdot\right)$ defines the copula density function. Next, by considering the Gaussian copula, \eqref{eq-gen-pdf} is rewritten as\vspace{0cm}
	\begin{multline}\label{eq-gb-pdf}
	f_{\gamma_\mathrm{b}}\left(\gamma_\mathrm{b}\right)=\frac{\left[f_{A^2}\left(\frac{\gamma_\mathrm{b}}{\overline{\gamma}_\mathrm{b}}\right)\right]^{N_\mathrm{b}}}{\overline{\gamma}_\mathrm{b}}\\
	\times\frac{\exp\left(-\frac{1}{2}\left(\boldsymbol{\mathbf{\varphi}}^{-1}_{A^2}\right)^T\left(\mathbf{R}_\mathrm{b}^{-1}-\mathbf{I}\right)\boldsymbol{\mathbf{\varphi}}^{-1}_{A^2}\right)}{\sqrt{{\rm det}\left(\mathbf{R}_\mathrm{b}\right)}},
\end{multline}
where $\mathrm{det}\left(\mathbf{R}_\mathrm{b}\right)$ defines the determinant of the correlation matrix $\mathbf{R}_\mathrm{b}$, $\mathbf{I}$ is the identity matrix, and $\boldsymbol{\mathbf{\varphi}}^{-1}_{A^2}=\left[\varphi^{-1}\left(u_j\right),\dots,\varphi^{-1}\left(u_j\right)\right]^T$. Next, by applying $f_{A^2}\left(a\right)$ from \eqref{eq-pdf-a} into \eqref{eq-gb-pdf}, \eqref{eq-pdf-gb} is derived and the proof is completed. 
	\end{proof}
	\begin{figure}
		\centering
		\vspace{0cm}
		\subfigure[]{%
			\includegraphics[width=0.25\textwidth]{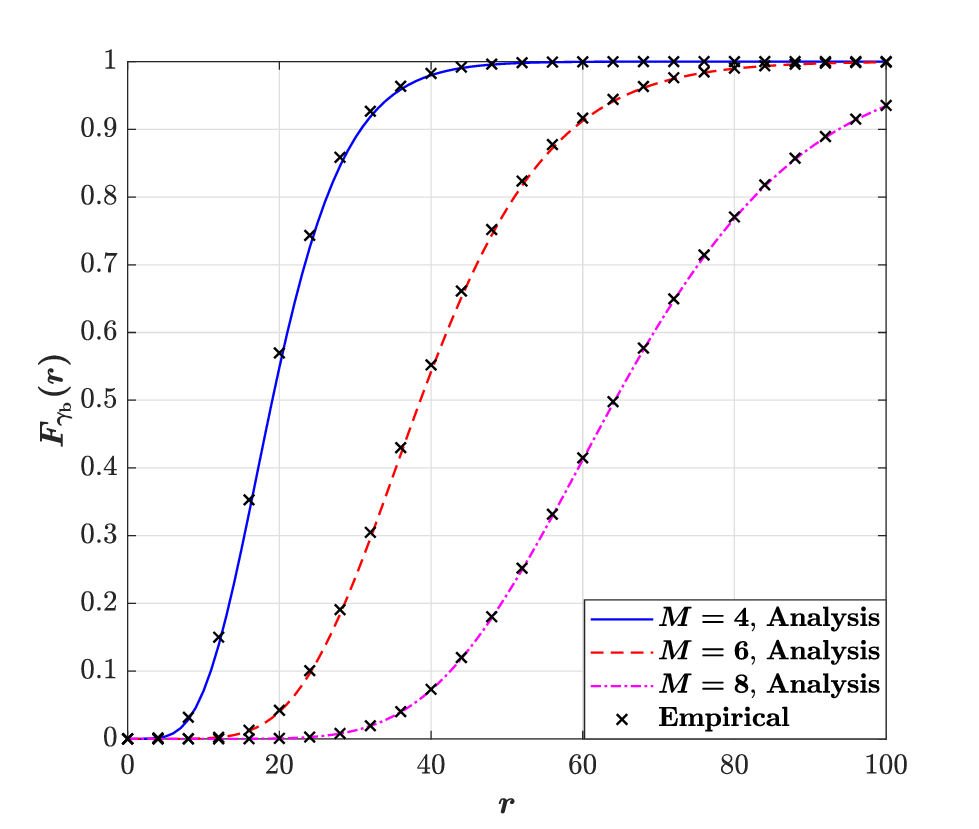}\label{fig-cdf}%
		}\vspace{0cm}
		\subfigure[]{%
			\includegraphics[width=0.25\textwidth]
			{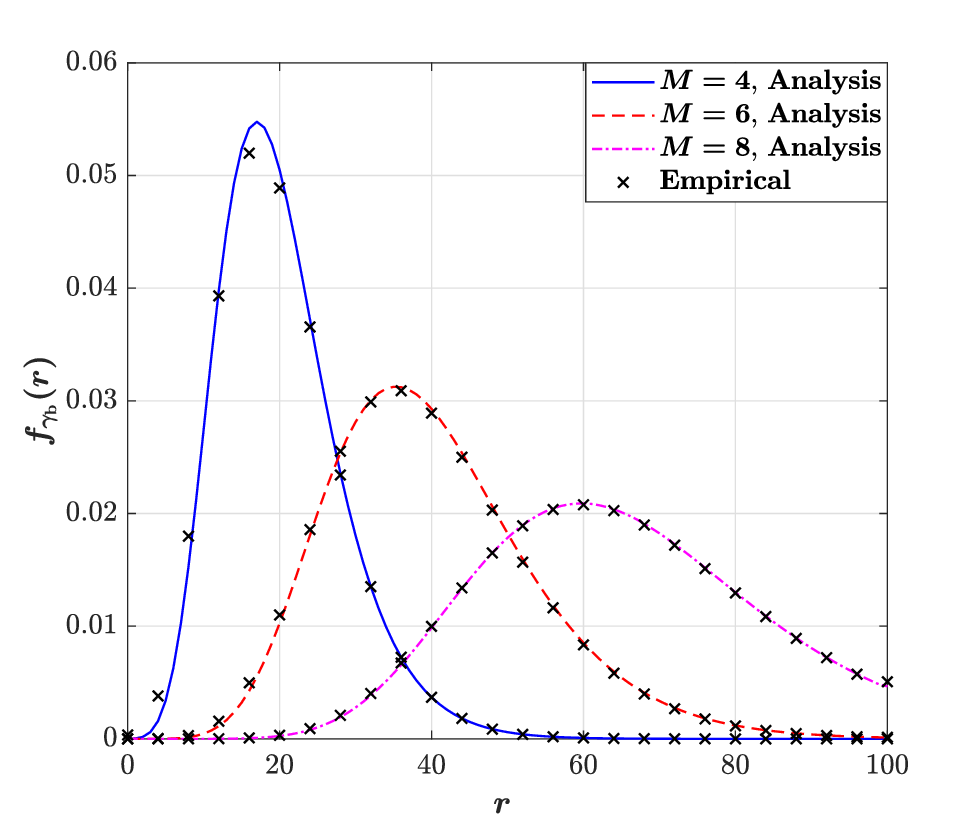}\label{figpdf}%
		}\\
		\vspace{0cm}
		\subfigure[]{%
			\includegraphics[width=0.25\textwidth]{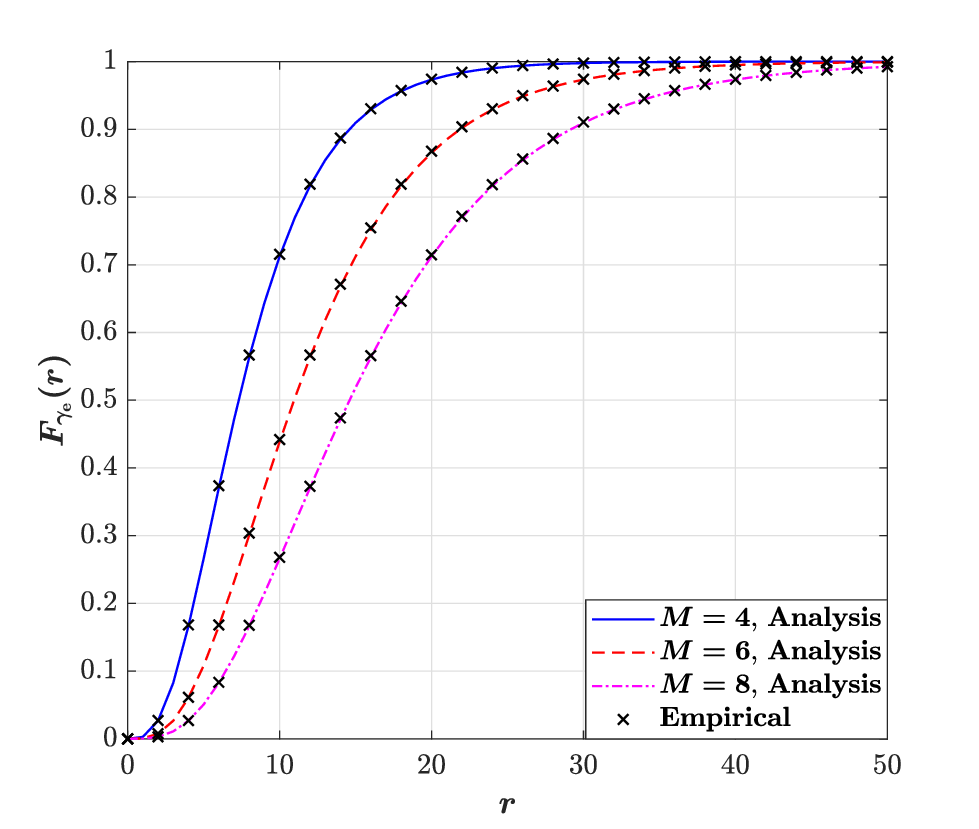}\label{fig-cdf_e}%
		}\vspace{0cm}
				\subfigure[]{%
			\includegraphics[width=0.25\textwidth]{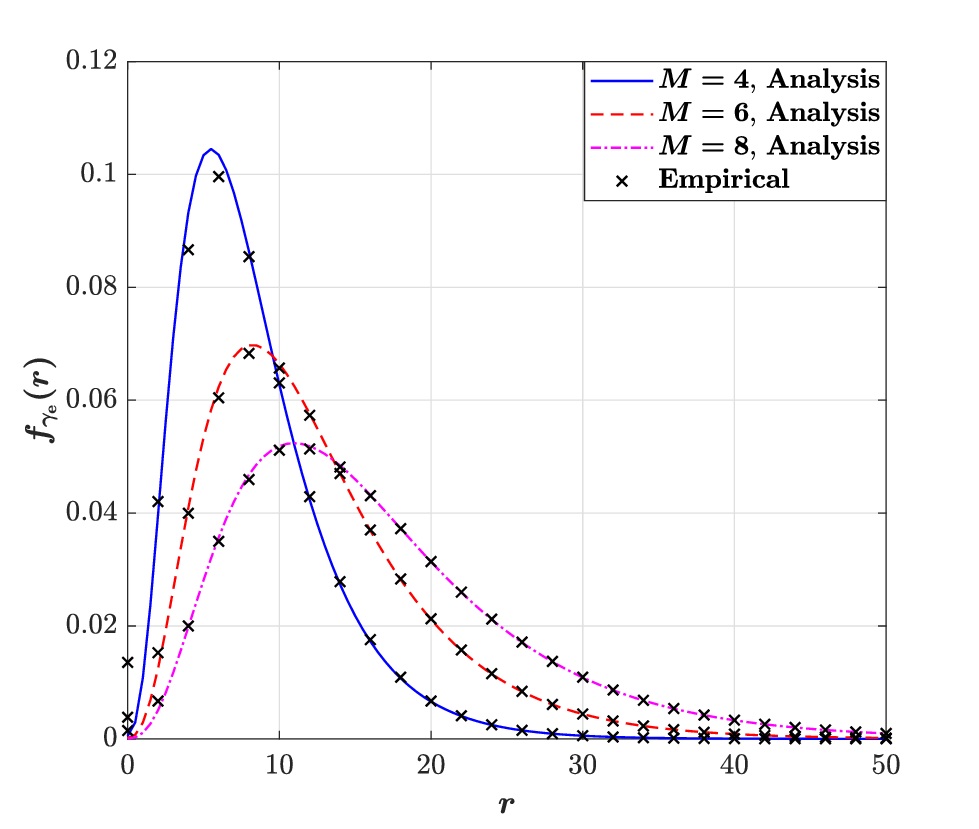}\label{fig-pdf_e}%
		}
		\vspace{-0.4cm}
		\caption{Analytical CDF and PDF of $\gamma_i$ for different $M$ and selected values of $N_i=4$ and $W_i=1\lambda^2$: (a) CDF of $\gamma_\mathrm{b}$, (b) PDF of $\gamma_\mathrm{b}$, (c) CDF of $\gamma_\mathrm{e}$, and (d) PDF of $\gamma_\mathrm{e}$. }\label{fig-dist}\vspace{-0.5cm}
	\end{figure}
\subsubsection{Distribution of $\gamma_\mathrm{e}$} Given  \eqref{eq-hfas} and \eqref{eq-snr}, the equivalent channel gain at $n_\mathrm{e}$-th port Eve can be extended as
\begin{align}
	\left|h_{\mathrm{eq},n_\mathrm{e}}\right|^2=\left|\sum_{m=1}^Mg_mh_{m,n_\mathrm{e}}^\mathrm{e}\mathrm{e}^{-j(\psi_m-\theta_m-\beta_m^\mathrm{e})}\right|^2=B^2,
\end{align}
where $\theta_m$ and $\beta_m^\mathrm{e}$ are uniformly distributed in $\left[0,2\pi\right)$. Given that the zero-mean RV $h_{m,n_\mathrm{e}}\mathrm{e}^{j\beta_m^\mathrm{e}}$ is independent of $g_m$, $\theta_m$, and $\psi_m$, we have $\mathbb{E}\left[g_mh_{m,n_\mathrm{e}}^\mathrm{e}\mathrm{e}^{-j(\psi_m-\theta_m-\beta_m^\mathrm{e})}\right]=0$ and $\mathrm{Var}\left[g_mh_{m,n_\mathrm{e}}^\mathrm{e}\mathrm{e}^{-j(\psi_m-\theta_m-\beta_m^\mathrm{e})}\right]=1$. Therefore, by utilizing the CLT for large values of $M$, $B$ is approximated a Gaussian RV with mean $\mu_B=0$ and variance $\sigma_B^2=M$. Hence, $B^2$ is an exponentially distributed RVs with the mean of $M$ and the following CDF and PDF, respectively \cite{shi2022secure}
	\begin{align}
		F_{B^2}\left(b\right)=1-\exp\left(-\frac{b}{M}\right), \label{eq-cdf-b}
	\end{align}
	\begin{align}
		f_{B^2}\left(b\right)=\frac{1}{M}\exp\left(-\frac{b}{M}\right).\label{eq-pdf-b}
	\end{align}
Now, since we have the distribution of the equivalent channel at Eve, the distribution of $\gamma_\mathrm{e}$ is determined as the following proposition. 
\begin{Proposition}The CDF and PDF of $\gamma_\mathrm{e}$ are given by \eqref{eq-cdf-ge} and \eqref{eq-pdf-ge}, respectively, where $\boldsymbol{\mathbf{\varphi}}^{-1}_{B^2}$ is defined in \eqref{eq-phi2} (see the top of the next page).
	\begin{figure*}[t]
		\normalsize
		\begin{align}
			&F_{\gamma_\mathrm{e}}\left(\gamma_\mathrm{e}\right)=\Phi_{\mathbf{R}_\mathrm{e}}\left(\sqrt{2}\mathrm{erf}^{-1}\left(1-2\exp\left(\frac{\gamma_\mathrm{e}}{M\overline{\gamma}_\mathrm{e}}\right)\right),\dots,\sqrt{2}\mathrm{erf}^{-1}\left(1-2\exp\left(\frac{\gamma_\mathrm{e}}{M\overline{\gamma}_\mathrm{e}}\right)\right);\vartheta_\mathrm{e}\right)\label{eq-cdf-ge}
		\end{align}
		\hrulefill
		\begin{align}\label{eq-pdf-ge}
			f_{\gamma_\mathrm{e}}\left(\gamma_\mathrm{e}\right)=\frac{1}{\overline{\gamma}_\mathrm{e}}\left[\frac{1}{M}\exp\left(-\frac{\gamma_\mathrm{e}}{M\overline{\gamma}_\mathrm{e}}\right)\right]^{N_\mathrm{e}}
			\frac{\exp\left(-\frac{1}{2}\left(\boldsymbol{\mathbf{\varphi}}^{-1}_{B^2}\right)^T\left(\mathbf{R}_\mathrm{e}^{-1}-\mathbf{I}\right)\boldsymbol{\mathbf{\varphi}}^{-1}_{B^2}\right)}{\sqrt{{\rm det}\left(\mathbf{R}_\mathrm{e}\right)}}
		\end{align}
		\hrulefill
		\begin{align}\label{eq-phi2}
			\boldsymbol{\mathbf{\varphi}}^{-1}_{B^2}=\left[\sqrt{2}\mathrm{erf}^{-1}\left(1-2\exp\left(\frac{\gamma_\mathrm{e}}{M\overline{\gamma}_\mathrm{e}}\right)\right),\dots,\sqrt{2}\mathrm{erf}^{-1}\left(1-2\exp\left(\frac{\gamma_\mathrm{e}}{M\overline{\gamma}_\mathrm{e}}\right)\right);\vartheta_\mathrm{e}\right]^T
		\end{align}
		\hrulefill
	\end{figure*}
	\end{Proposition}
\begin{proof}
The details of proof are similar to the proof of Proposition 1, where the CDF and PDF of $B^2$ and the correlation matrix $\mathbf{R}_\mathrm{e}$ need to be applied into \eqref{eq-cdf-g} and \eqref{eq-gen-pdf}. 
\end{proof}
Fig. \ref{fig-dist} compares the derived analytical CDF and PDF of the SNR $\gamma_i$, provided in \eqref{eq-cdf-gb}, \eqref{eq-pdf-gb}, \eqref{eq-cdf-ge}, and \eqref{eq-pdf-ge}, with Monte Carlo simulations. It can be observed that the analytical and simulation results are in perfect agreement, confirming the accuracy of the derived expressions using the CLT and Gaussian copula function. Furthermore, we can see that increasing $M$ significantly improves the fading conditions.
\subsection{SOP Analysis}
The SOP is defined as the probability that the instantaneous secrecy capacity $\mathcal{C}_\mathrm{s}$ is less than a target secrecy rate $R_\mathrm{s}$, i.e.,
\begin{align}
	P_\mathrm{sop}=\Pr\left(\mathcal{C}_\mathrm{s}\left(\gamma_\mathrm{b},\gamma_\mathrm{e}\right)\le R_\mathrm{s}\right),
\end{align}
in which, from information-theoretic viewpoint, the secrecy capacity is defined as the difference between the channel capacities corresponding to legitimate and wiretap links, i.e.,
\begin{align}
	\mathcal{C}_\mathrm{s}\left(\gamma_\mathrm{b},\gamma_\mathrm{e}\right)=\max\left\{\log_2\left(1+\gamma_\mathrm{b}\right)-\log_2\left(1+\gamma_\mathrm{e}\right),0\right\},\label{eq-cs2}
\end{align}
where $\mathcal{C}_i=\log_2\left(1+\gamma_i\right)$ denotes the capacity for the main (Alice to Bob) and eavesdropper (Alice to Eve) channels.
\begin{Proposition}
The SOP of the considered RIS-aided FAS is given by \eqref{eq-sop}, where $R_\mathrm{o}=2^{R_\mathrm{s}}$ and $R_\mathrm{t}=R_\mathrm{o}-1$.
\begin{figure*}
	\normalsize
	\begin{multline}
	P_\mathrm{sop}\approx\sum_{k=1}^K \frac{\omega_k\mathrm{e}^{\epsilon_k}}{\overline{\gamma}_\mathrm{e}}\left[\frac{1}{M}\exp\left(-\frac{\epsilon_k}{M\overline{\gamma}_\mathrm{e}}\right)\right]^{N_\mathrm{e}}
	\frac{\exp\left(-\frac{1}{2}\left(\boldsymbol{\mathbf{\varphi}}^{-1}_{B^2}\right)^T\left(\mathbf{R}_\mathrm{e}^{-1}-\mathbf{I}\right)\boldsymbol{\mathbf{\varphi}}^{-1}_{B^2}\right)}{\sqrt{{\rm det}\left(\mathbf{R}_\mathrm{e}\right)}}\\
	\times\Phi_{\mathbf{R}_\mathrm{b}}\left(\sqrt{2}\mathrm{erf}^{-1}\left(1-2Q_{\frac{1}{2}}\left(\sqrt{\frac{\tau}{\sigma_A^2}},\sqrt{\frac{R_\mathrm{o}\epsilon_k+R_\mathrm{t}}{\sigma_A^2\overline{\gamma}_\mathrm{b}}}\right)\right),\dots,\sqrt{2}\mathrm{erf}^{-1}\left(1-2Q_{\frac{1}{2}}\left(\sqrt{\frac{\tau}{\sigma_A^2}},\sqrt{\frac{R_\mathrm{o}\epsilon_k+R_\mathrm{t}}{\sigma_A^2\overline{\gamma}_\mathrm{b}}}\right)\right);\vartheta_\mathrm{b}\right)\label{eq-sop}
	\end{multline}\vspace{-0.6cm}
	\hrulefill
\end{figure*}
\end{Proposition} 
\begin{proof} Given \eqref{eq-cs2} and the  definition of SOP, we have
	\begin{align}
		&P_\mathrm{sop} = \Pr\left(\log_2\left(\frac{1+\gamma_\mathrm{b}}{1+\gamma_\mathrm{e}}\leq R_\mathrm{s}\right)\right)\\
		&=\Pr\left(\gamma_\mathrm{b}\leq2^{R_\mathrm{s}}\gamma_\mathrm{e}+2^{R_\mathrm{s}}-1\right)\\
		&=\int_0^\infty F_{\gamma_\mathrm{b}}\left(R_\mathrm{o}\gamma_\mathrm{e}+R_\mathrm{t}\right)f_{\gamma_\mathrm{e}}\left(\gamma_\mathrm{e}\right)\mathrm{d}\gamma_\mathrm{e},\label{eq-sop1}
	\end{align}
where $R_\mathrm{o}=2^{R_\mathrm{s}}$ and $R_\mathrm{t}=R_\mathrm{o}-1$. By inserting the corresponding CDF and PDF from  \eqref{eq-cdf-gb} and \eqref{eq-pdf-ge} into \eqref{eq-sop1}, it can be seen that computing the integral in \eqref{eq-sop1} is almost mathematically intractable due to complexity of the integrand. In this regard, the integral can be solved numerically by using programming language such as MATLAB or numerical integral techniques such GLQ. 
\begin{lemma}\label{eq-lemma1}
	The GLQ is defined as
	\begin{align}
		\int_0^\infty\mathrm{e}^{-x}\Lambda\left(x\right)\mathrm{d}x\approx\sum_{k=1}^{K}\omega_k\Lambda\left(\epsilon_k\right)
	\end{align}
	where $\epsilon_k$ is the $k$-th root of Laguerre polynomial $L_{K}\left(\epsilon_k\right)$ with the weight $\omega_k$.
\end{lemma}
Now, by applying Lemma \ref{eq-lemma1} into \eqref{eq-sop1}, the proof is accomplished.
\end{proof}
\section{Numerical Results}
Here, we evaluate the proposed system's performance through our analytical derivations, which are validated by Monte Carlo simulations. The simulation parameters are set as follows:  $R_\mathrm{s}=3$ bits, $d_\mathrm{r}=100$ m, $d_\mathrm{b}=500$ m, $d_\mathrm{e}=800$ m, $\alpha=2.1$, $\sigma^2_\mathrm{b}=-70$ dBm, $\sigma^2_\mathrm{e}=-50$ dBm, $K=2$, $\epsilon_k=\left\{2-\sqrt{2}, 2+\sqrt{2}\right\}$, and $\omega_k=\left\{0.8536, 0.1464\right\}$. Furthermore, we utilize Algorithm 1 in \cite{ghadi2023gaussian}, to implement the Gaussian copula function, which includes the joint CDF and PDF of the normal distribution. 

Fig. \ref{fig-sop} illustrates the performance of SOP in terms of the average SNR at Bob $\overline{\gamma}_\mathrm{b}$ for different values of fluid antenna size $W_\mathrm{b}$ and the number of fluid antenna ports $N_\mathrm{b}$. As expected, it can be observed that increasing $\gamma_\mathrm{b}$ improves the SOP performance since the main channel condition becomes better. Moreover, we can see that by simultaneously increasing $W_\mathrm{b}$ and $N_\mathrm{b}$, the SOP value remarkably decreases since the spatial correlation between the FAS ports becomes balanced. Therefore, the diversity gain and spatial multiplexing improve which lead to better signal quality and more reliable communications. For instance, for a given $\overline{\gamma}_\mathrm{b}=2$ dB, the SOP is in the order of $10^{-5}$ with $N_\mathrm{b}=4$ and $W_\mathrm{b}=1\lambda^2$, while  the SOP achieves the order of $10^{-8}$ for $N_\mathrm{b}=8$ and $W_\mathrm{b}=3\lambda^2$. Additionally, by using the single fixed-antenna system as a benchmark, we can see that the FAS with only activated port performs better in terms of SOP compared to the TAS.

Fig. \ref{fig-sop_m} represents the SOP behavior versus $\overline{\gamma}_\mathrm{b}$ for different numbers of RIS elements $M$. As $M$ increases, we observe a significant improvement in SOP performance. This improvement is due to the enhanced signal control and directionality, increased path diversity and array gain, and boosted secrecy capacity that comes with more RIS elements. Consequently, more robust countermeasures against eavesdropping are implemented, resulting in a more secure transmission. Moreover, we observe that increasing $M$ has a greater impact on reducing SOP in the FAS compared to TAS.  

To gain more insights into the impact of fluid antenna size and the number of fluid antenna ports on the SOP performance, we present Fig. \ref{fig-sop2}. Participially, Fig. \ref{fig-sop_w} indicates the SOP performance in terms of $W_\mathrm{b}$ for a fixed $W_\mathrm{b}$ and different values of $\overline{\gamma}_\mathrm{b}$. It can be observed that the SOP decreases without bound as $W_\mathrm{b}$ grows for sufficient large value of $N_\mathrm{b}$. This improvement is expected because increasing the size of the fluid antenna, while keeping the number of ports both fixed and large, significantly reduces spatial correlation, which leads to better signal reception, reduced interference, and ultimately, a lower SOP. Furthermore, we can see this improvement becomes more noticeable for larger $\overline{\gamma}_\mathrm{b}$. Fig. \ref{fig-sop_n} illustrates the performance of SOP versus $N_\mathrm{b}$ for a given $W_\mathrm{b}$ and different $\overline{\gamma}_\mathrm{b}$. We can observe that the SOP performance initially enhances and then becomes saturated as $N_\mathrm{b}$ continuously increases. This occurs primarily because increasing $N_\mathrm{b}$ while keeping $W_\mathrm{b}$ fixed causes the ports to become increasingly clustered. As a result, spatial correlation effects become more pronounced, diminishing the diversity gain after a certain threshold is reached. Consequently, the reduction in SOP slows and eventually reaches a saturation point.   
	\begin{figure}
	\centering
	\vspace{-0cm}
	\subfigure[]{%
		\includegraphics[width=0.25\textwidth]{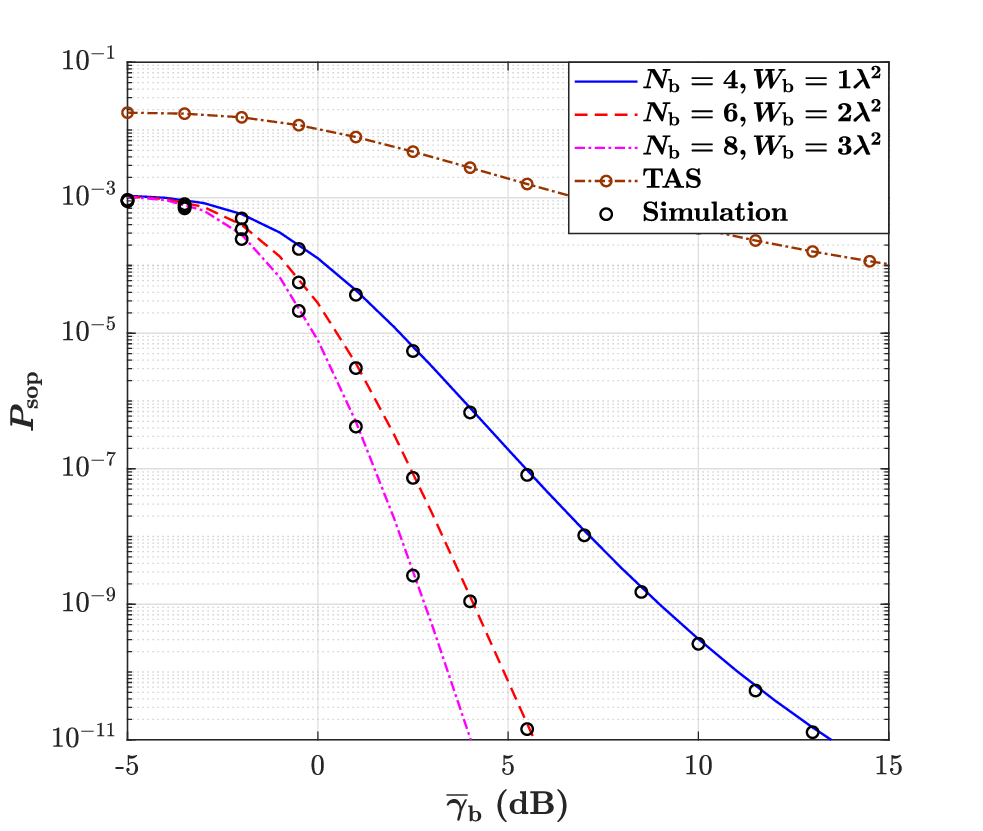}\label{fig-sop}%
		}\vspace{-0.35cm}
	\subfigure[]{%
		\includegraphics[width=0.25\textwidth]{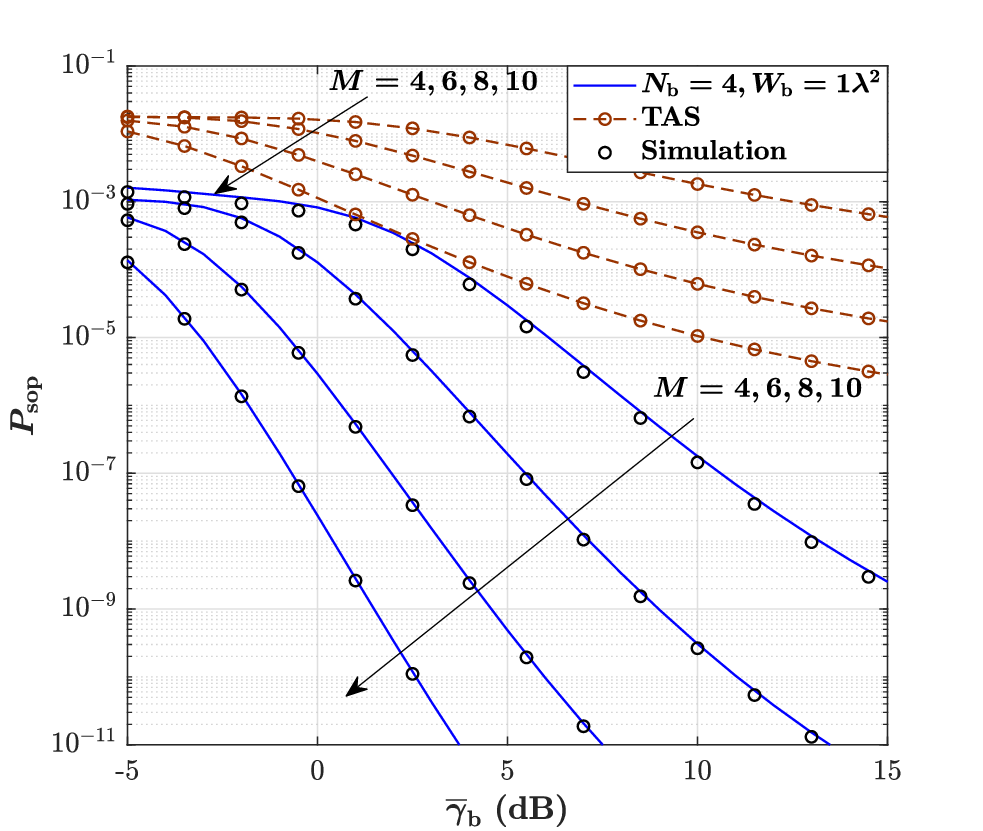}\label{fig-sop_m}%
	}
	\vspace{-0.1cm}
	\caption{(a) SOP versus $\overline{\gamma}_\mathrm{b}$ for different values of $N_\mathrm{b}$ when $M=6$, $N_\mathrm{e}=4$, and $W_\mathrm{e}=1\lambda^2$. (b) SOP versus $\overline{\gamma}_\mathrm{b}$ for different values of $M$ when $N_i=4$ and $W_i=1\lambda^2$.}\label{fig-sop1}\vspace{-0.5cm}
\end{figure}
	\begin{figure}
	\centering
	\vspace{-0.3cm}
	\subfigure[]{%
		\includegraphics[width=0.25\textwidth]{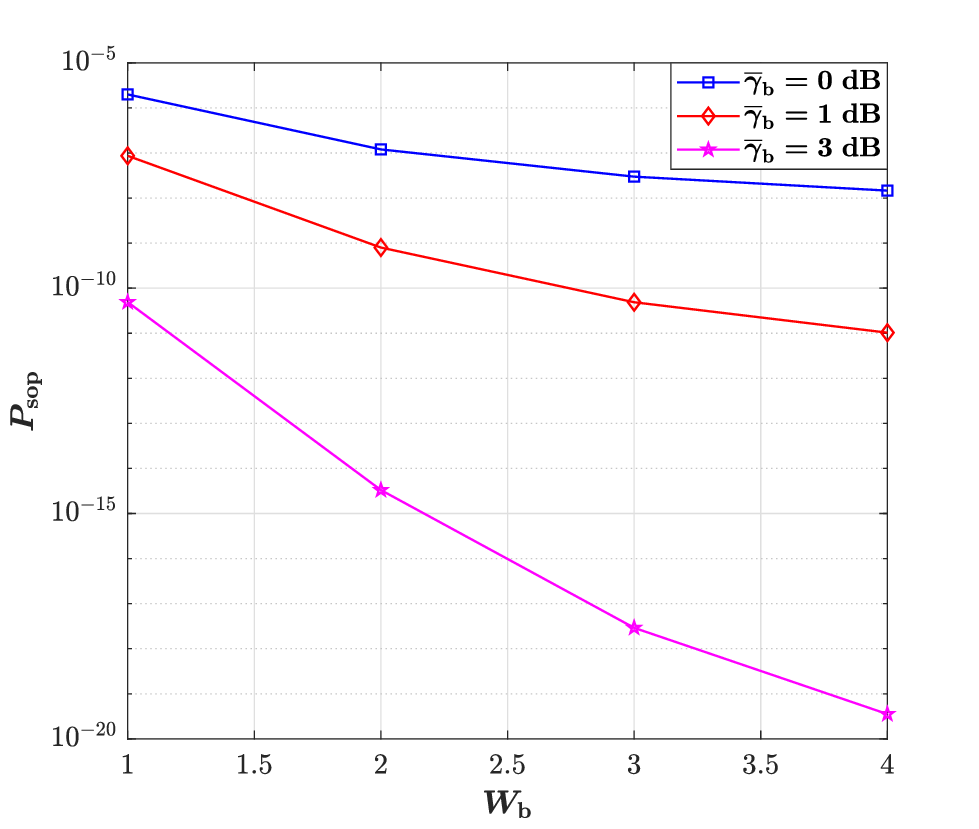}\label{fig-sop_w}%
	}
	\subfigure[]{%
		\includegraphics[width=0.25\textwidth]{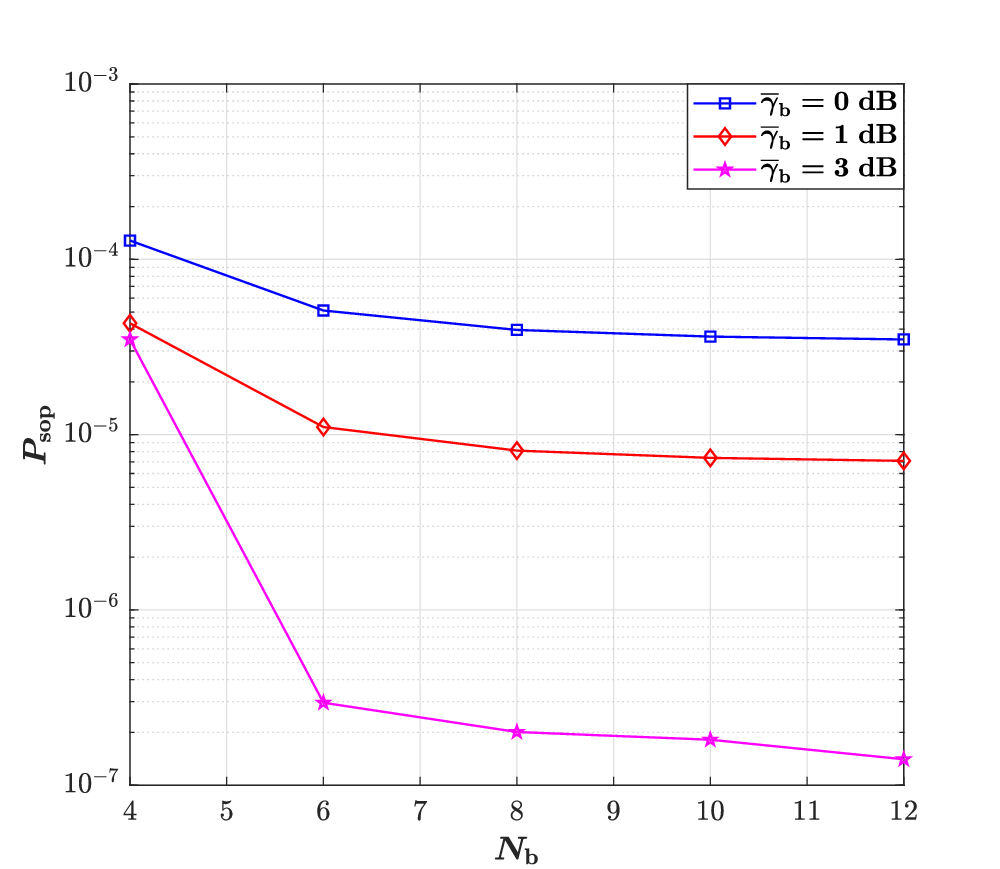}\label{fig-sop_n}%
	}
	\vspace{-0.5cm}
	\caption{(a) SOP versus $W_\mathrm{b}$ for a fixed $N_\mathrm{b}=16$ and (b) SOP versus $N_\mathrm{b}$ for a fixed $W_\mathrm{b}=1\lambda^2$ for selected values of $\bar{\gamma}_\mathrm{b}$, when $M=6$, $N_\mathrm{e}=4$, and $W_\mathrm{e}=1\lambda^2$.}\label{fig-sop2}\vspace{-0.5cm}
\end{figure}\vspace{-0.6cm}
\section{Conclusion}\vspace{-0.4cm}
In this paper, we investigated the secrecy performance of FAS in a RIS-aided secure communication system. Specifically, by considering a classic wiretap channel, we assumed that Alice sends a confidential message to Bob through a RIS, while Eve attempts to decode the information from her received signal. Additionally, we assumed that Alice has a fixed antenna, whereas both Bob and Eve are equipped with FAS. In this context, we first derived the CDF and PDF of the received SNR at Bob and Eve using the CLT and the Gaussian copula function. Next, we obtained the compact analytical expression of the SOP by employing the numerical integration GLQ method. Our numerical results demonstrated the advantages of using FAS in RIS-aided communication, providing more secure and reliable transmission.\vspace{-0.2cm}
\section*{Acknowledgment}
\small This work is supported by the Engineering and Physical Sciences Research Council (EPSRC) under Grant EP/W026813/1, the European Union's Horizon 2022 Research and Innovation Programme under Marie Sk\l{}odowska-Curie Grant No. 101107993, the grant PID2023-149975OB-I00 (COSTUME) funded by MICIU/AEI/10.13039/501100011033, and the EMERGIA20-00297 grant funded by Junta de Andaluc\'ia.
\vspace{-0.1cm}



\end{document}